\documentclass{article}

\usepackage{arxiv}

\usepackage[utf8]{inputenc} 
\usepackage[T1]{fontenc}    
\usepackage{hyperref}       
\usepackage{url}            
\usepackage{booktabs}       
\usepackage{amsfonts}       
\usepackage{nicefrac}       
\usepackage{microtype}      
\usepackage{lipsum}		
\usepackage{graphicx}
\usepackage{doi}
\usepackage{tikz}
\usepackage{amssymb}

\usepackage{amsthm}
\usepackage{amsmath}
\usepackage{comment}
\usepackage{soul}
\newcommand{\cmt}[1]{}
\usepackage{color}
\usepackage{bm}
\usepackage{lineno}

\newtheorem{definition}{Definition}

\newtheorem{theorem}{Theorem}
\newtheorem{lemma}{Lemma}

\theoremstyle{definition}

\DeclareMathOperator*{\argmax}{arg\,max}

\title{On the Oscillations in Cournot Games with Best Response Strategies}

\author{ \href{https://orcid.org/0000-0002-1760-3305}{Zhengyang Liu} \\
	Beijing Institute of Technology\\
	Beijing, China 100081 \\
	\texttt{zhengyang@bit.edu.cn} \\
	\And
	\href{https://orcid.org/0009-0008-0174-7604}{Haolin Lu} \\
	Beijing Institute of Technology\\
	Beijing, China 100081 \\
	\texttt{haolin@bit.edu.cn}\\
        \And
        \href{https://orcid.org/0009-0001-6558-812X}{Liang Shan} \\
	Renmin University of China\\
	Beijing, China 100872 \\
	\texttt{shanliang@ruc.edu.cn} \\
	\And
        \href{https://orcid.org/0000-0001-7752-4687}{Zihe Wang}\thanks{Corresponding author} \\
	Renmin University of China\\
	Beijing, China 100872 \\
	\texttt{wang.zihe@ruc.edu.cn}
}

\begin{document}
\maketitle

\begin{abstract}
In this paper, we consider the dynamic oscillation in the Cournot oligopoly model, which involves multiple firms producing homogeneous products. To explore the oscillation under the updates of best response strategies, we focus on the linear price functions. In this setting, we establish the existence of oscillations. In particular, we show that for the scenario of different costs among firms, the best response converges to either a unique equilibrium or a two-period oscillation. We further characterize the oscillations and propose linear-time algorithms for finding all types of two-period oscillations. To the best of our knowledge, our work is the first step toward fully analyzing the periodic oscillation in the Cournot oligopoly model.
\end{abstract}


\section{Introduction}
 
The Cournot competition~\cite{cournot1838researches} is one of the most fundamental economic models in the oligopoly theory. 
In the classic Cournot game, $n$ firms produce homogeneous products. Each firm has market power, as its output affects the market price, and aims to maximize its profits. All the firms are assumed to utilize the production quantities as strategic variables, which they decide simultaneously and independently of each other. The existence of equilibrium has been demonstrated in a broad range of models~\cite{debreu1952social,szidarovszky1977new}, while the question of whether firms can automatically converge to equilibrium through best response strategies remains a challenging problem.

When the firms update their strategies asynchronously, they ultimately reach the Nash equilibrium since the Cournot game can be considered as a potential game~\cite{monderer1996potential}. The scenario becomes more intriguing when the firms {\em synchronize} their strategy updates. That is, each firm chooses the best response strategy, while others keep their strategies (production quantities) unchanged. 
Palander and Theocharis~\cite{palander1939konkurrens,theocharis1960stability} showed that convergence depends on the number of firms in the game.
They demonstrated that equilibrium is attained with two firms, but with three firms, a stationary oscillation ensues. Beyond three firms, significant instability arises. 
Puu~\cite{puu1996complex} conducted numerous simulations of the firms' actions.
This line of research~\cite{mcmanus1961comments,fisher1961stability,chrysanthopoulos2021adaptive} considered the properties of local dynamics, assuming the firms' quantities are always non-negative. Consequently, the dynamics of their strategies can be described by the system of linear equations. 
C{\'a}novas et al.~\cite{canovas2008cournot} pioneered the examination of global dynamics, considering the constraint that supply quantities must be non-negative. Thus, the dynamics of their strategies are delineated by a set of non-linear equations. They identified specific scenarios where periodic orbits of period two exist within the dynamics.
Furthermore, C{\'a}novas et al.~\cite{canovas2009reducing,canovas2022monopoly} delved into circumstances where the number of firms diminishes to a monopoly within the dynamics.
To the end, all the previous works typically first assume a particular type of oscillation with a period of two, followed by establishing conditions under which such oscillation is feasible based on the costs incurred by the firms.
Then a natural question arises: 
\begin{quote}
{\em Does an oscillation of a longer period exist in a general setting with multiple firms?}
\end{quote}

In this work, we address this issue and fully characterize the oscillations in the Cournot oligopoly model.
Our model builds upon the foundation laid by C{\'a}novas et al.~\cite{canovas2008cournot}. 
We examine a Cournot oligopoly model with a population of 
$n$ firms producing homogeneous products. All firms make decisions simultaneously, considering a linear price function. Production costs may vary among firms, and we impose the constraint that production quantities cannot be negative.
The contributions of this paper are summarized as follows:
\begin{itemize}
    \item Under linear inverse demand with non‑negative quantity constraints, we prove that simultaneous best‑response dynamics in an $n$-firm Cournot game must converge exclusively to either a Nash equilibrium or a two‑period cycle.
    \item We further establish that any two-period cycle must conform to one of three specific structural patterns. The first type arises when each firm produces zero output in one of the two periods. The second type occurs when fewer than three firms produce a positive quantity in the first period, with additional firms producing positive quantities in the second period. The third type arises in the case of exactly three firms.
    \item We propose linear-time $O(n)$ algorithms that simultaneously detect oscillation patterns and classify their structural type. This computational contribution equips practitioners with a diagnostic toolkit for empirically analyzing strategic cycling behavior in real-world oligopolistic markets.
\end{itemize}

\subsection{Related work}


In his seminal work, Cournot~\cite{cournot1838researches} introduced a mathematical model to characterize the features of duopoly markets. Since then, there has been considerable interest in whether firms could reach equilibrium through best response strategies, leading to numerous studies on equilibrium in the Cournot oligopoly model.
Based on Cournot's model, Palander~\cite{palander1939konkurrens} discovered that convergence fails when the number of competitors exceeds three. When there are three competitors, the Cournot equilibrium tends towards endless but stationary oscillation. 
Theocharis~\cite{theocharis1960stability} independently studied the Cournot model and obtained results similar to Palander's. 
Numerous of studies~\cite{puu1996complex,puu1998chaotic,agiza1998explicit,ahmed1998dynamics} investigated the impact of increasing the number of firms in the Cournot model on equilibrium outcomes. C{\'a}novas et al.~\cite{canovas2009reducing,canovas2022monopoly} explored conditions leading to the reduction in the number of competitors in the Cournot oligopoly model to duopoly or monopoly. Our work extends this line by thoroughly examining the periodic oscillations in Cournot competition with multiple firms, addressing the conditions under which these oscillations occur and characterizing their nature.

The significant influence of adjustment rules on the convergence to equilibrium within the Cournot oligopoly model has been extensively studied.
Fisher~\cite{fisher1961stability} highlighted the pivotal role of adjustment speed in facilitating convergence in Cournot games involving multiple firms. Building upon the aforementioned insights, Nowaihi and Levine~\cite{al1985stability} explored scenarios involving both discrete and continuous adjustments. Okuguchi~\cite{okuguchi1970adaptive} delved into models where firms employ different adjustment rules, assuming that the $i$-th firm's expectation of the $k$-th firm's output is not necessarily the same as the $j$-th firm's. 
Matsumoto \cite{matsumoto2014theocharis} investigated production differentiation oligopoly, focusing on adjustments in pricing and quantities.
Furthermore, subsequent studies have conducted experiments examining adaptation rules within the Cournot game~\cite{huck2002stability,huck1999learning,rassenti2000adaptation,cox1998learning}. The best response strategy is not only widely applied in the field of Cournot games, but also extensively studied in various other related areas, such as lottery contests~\cite{ghosh2023best}, continuous zero-sum games~\cite{hofbauer2006best} and additive aggregation finite games~\cite{kukushkin2004best}. Our contribution builds upon these studies by providing a detailed classification of oscillation types and introducing linear-time algorithms to identify these oscillations, thereby offering new insights into the dynamic behavior of firms in Cournot competition.


In another research line, Ahmed et al.~\cite{ahmed2000modifications} began to relax the assumption in Puu's work~\cite{puu1996complex} that each firm knows the output of its competitors and is aware of the market demand function. These studies focused on bounded rationality, where firms lack complete information about the market demand function. Bischi and Naimzada~\cite{bischi2000global} presented a dynamic Cournot model with bounded rationality. Following this, a series of works by Agiza~\cite{agiza2001dynamics,agiza2002complex,agiza2003nonlinear,agiza2004chaotic} explored the bounded rationality with firms' strategy updates by using the marginal profit approach. 
The complex dynamics, bifurcations, and chaos are observed in the numerical simulations~\cite{zhang2007analysis,fan2012complex}.
Even-Dar~\cite{even2009convergence} showed that using gradient-based no-external regret procedures can guarantee convergence to equilibrium.
Our research further differentiates itself by examining the constraints on non-negative production quantities.

There are also discussions regarding how the firms' adjustments, whether simultaneous or sequential, affect equilibrium outcomes. Tullock~\cite{tullock2001efficient} constructed a rent-seeking model and proposed the possibility of an ``Intellectual mire'' situation that could lead to the non-existence of equilibrium. Leininger~\cite{leininger1993more} later addressed the issues raised by Tullock using sequential play methods. Subsequently, a series of articles studied the dynamics of the first or second move in the duopoly model~\cite{gal1985first,reinganum1985two}. Building upon this foundation, these cases where more than two firms sequentially take action were further investigated~\cite{puu1996complex,puu1998chaotic}. Our work specifically focuses on simultaneous updates in a Cournot setting with multiple firms, exploring the resulting oscillations.

\section{Preliminaries}
We consider a Cournot oligopoly model, where a population of $n$ firms in the market produces homogeneous products. All firms make simultaneous decisions about production quantities in each round. At round $t$, we denote by $q_i^t$ the output decision of firm $i$, and the market price $p^t$ is determined by the inverse demand function. In this work, we consider the {\em linear} price function as following~\cite{theocharis1960stability,bischi2000global,agiza2004chaotic,ding2014analysis,fiat2019beyond}:
\begin{align*}
    p^t := A  - \sum_{i = 1}^n q_i^t ,
\end{align*}
where $A$ is the market capacity which is a positive constant.

Suppose each firm can have different production technology, firm $i$ producing a unit item with the cost $c_i \ge 0$.
Therefore, the production cost for firm $i$ with producing $q_i$ items is $c_i q_i$.

The utility for firm $i$ at round $t$ with production quantity $q_i^t$ is: 
\begin{align*}
    u_i^t = (A - \sum_{j = 1}^n q_j^t - c_i) q_i^t.
\end{align*}

Assuming the update rules for all firms use the {\em best response} strategy, for $t\ge 1$, we have 
\begin{align}\label{BR_1}
    q_i^{t+1} = \argmax_{q_i \ge 0} \left\{q_i(A - \sum_{j \not =i}q_{j}^{t} -q_i -c_i) \right\}.
\end{align}

The R.H.S. of Eq.~\eqref{BR_1} is a singleton set since the objective is quadratic w.r.t. $q_i$, and we abuse the equality sign here.
Consider the production quantity strategy $q_i^{t+1}$ at round $t+1$ is maximized at $q_i^{t+1} = \frac{A - \sum_{j \not =i}q_{j}^{t} -c_i}{2}$, may lead to a negative quantity. However, a firm cannot produce a negative number of products. Hence, we further limit the lower bound of the production quantity~\cite{canovas2008cournot,canovas2022monopoly}. That is, for any $i\in[n]$,
\begin{align}\label{q_i^{t+1}}
    q_i^{t+1} := \max\left\{ 0,  \frac{A - \sum_{j \not =i}q_{j}^{t} -c_i}{2}  \right\}.
\end{align}
Next, we give the definition of the Nash equilibrium in our setting.
\begin{definition}[Nash equilibrium]
Given $\bm{q}^t = (q_1^t, q_2^t, \ldots, q_n^t)$ as a production quantity vector, we say it is a Nash equilibrium if it is unchanged under the update rule as Eq.~\eqref{q_i^{t+1}}. That is,
     \begin{align*}
         (q_1^{t+1}, q_2^{t+1}, \ldots, q_n^{t+1}) = (q_1^t, q_2^t, \ldots, q_n^t).
     \end{align*}
\end{definition}

\section{Periodic Oscillation}
In this section, we start to analyze the periodic oscillations. W.l.o.g., we assume that the firms are sorted according to the cost such that $0 \le c_1\le c_2\le\cdots\le c_n$. We first show that firms' production quantities are in an inverse relationship with the costs after a finite round of best response.   

\begin{lemma}\label{lemma1}
For any $i>j$ and any $t$, we let $Q^{t}=q_i^{t}-q_j^{t} $   and  $\Delta c=c_i-c_j\geq 0$. There exists an integer $T>0$ such that      
    $Q^{t}\leq0$ for any $t>T$. Moreover, if $\Delta c=0$, then  
 $\lim_{t\rightarrow \infty}Q^t=0$.
\end{lemma}

\begin{proof}

It follows from Eq.~\eqref{q_i^{t+1}} that
\begin{align*}
Q^{t+1}
= &\max\left\{ 0, \frac{A - \sum_{k \not =i}q_{k}^{t} -c_i}{2} \right\}  -\max\left\{ 0,  \frac{A - \sum_{k \not =j}q_{k}^{t} -c_j}{2}\right\} \\
= &\max\left\{ 0, \frac{A - \sum_{k \not =j}q_{k}^{t} -c_j }{2} + \frac{1 }{2}  Q^t
-\frac{ 1 }{2}\Delta c\right\} -\max\left\{ 0, \frac{A - \sum_{k \not =j}q_{k}^{t} -c_j}{2}  \right\} .
\end{align*}
Noting $\max\{0,x\}=\frac{1}{2}(x+|x|)$ and $\min\{0,x\}=\frac{1}{2}(x-|x|)$, we obtain
\begin{align*}
    \min\{0,y\}\leq\max\{0,x+y\}-\max\{0,x\}\leq \max\{0,y\}.
\end{align*}
Then we conclude 
\begin{align}\label{Q-q}
\min\left\{ 0, f(Q^t) \right\}\leq Q^{t+1}\leq  \max\left\{ 0, f(Q^t)\right\}, 
\end{align} 
where $f(Q^t)=\frac{1 }{2}  Q^t -\frac{ 1 }{2}\Delta c$. 
In Fig.~\ref{lemma1_fig}, we demonstrate the relationship between $Q^t$ and $Q^{t+1}$.
Since $\Delta c\geq 0$, the shadowed part of the figure is the feasible domain for $(Q^t, Q^{t+1})$.

\begin{figure}[t]
    \centering

\tikzset{every picture/.style={line width=0.75pt}} 

\begin{tikzpicture}[x=0.75pt,y=0.75pt,yscale=-1,xscale=1]

\draw    (89,200.28) -- (140.69,200.4) -- (377.1,200.95) ;
\draw [shift={(379.1,200.95)}, rotate = 180.13] [color={rgb, 255:red, 0; green, 0; blue, 0 }  ][line width=0.75]    (10.93,-3.29) .. controls (6.95,-1.4) and (3.31,-0.3) .. (0,0) .. controls (3.31,0.3) and (6.95,1.4) .. (10.93,3.29)   ;
\draw    (179.56,260.15) -- (180.56,107.81) ;
\draw [shift={(180.57,105.81)}, rotate = 90.38] [color={rgb, 255:red, 0; green, 0; blue, 0 }  ][line width=0.75]    (10.93,-3.29) .. controls (6.95,-1.4) and (3.31,-0.3) .. (0,0) .. controls (3.31,0.3) and (6.95,1.4) .. (10.93,3.29)   ;
\draw    (220.31,198.63) -- (220.31,200.85) ;
\draw    (98.67,260.28) -- (350.57,137.05) ;
\draw  [draw opacity=0][fill={rgb, 255:red, 0; green, 0; blue, 0 }  ,fill opacity=0.4 ] (356.48,134.36) -- (356.67,201.27) -- (222.29,200) -- cycle ;
\draw  [draw opacity=0][fill={rgb, 255:red, 0; green, 0; blue, 0 }  ,fill opacity=0.4 ] (104.96,256.57) -- (104.56,200.57) -- (220.31,200.85) -- cycle ;
\draw    (296,123) -- (312.37,152.78) ;
\draw [shift={(313.33,154.53)}, rotate = 241.2] [color={rgb, 255:red, 0; green, 0; blue, 0 }  ][line width=0.75]    (10.93,-3.29) .. controls (6.95,-1.4) and (3.31,-0.3) .. (0,0) .. controls (3.31,0.3) and (6.95,1.4) .. (10.93,3.29)   ;

\draw (276.04,91.49) node [anchor=north west][inner sep=0.75pt]  [font=\small] [align=left] {$\displaystyle f\left( Q^{t}\right) =\frac{1}{2} Q^{t} -\frac{1}{2} \Delta c$};
\draw (166.37,201.02) node [anchor=north west][inner sep=0.75pt]   [align=left] {0};
\draw (213.46,205.02) node [anchor=north west][inner sep=0.75pt]  [font=\scriptsize] [align=left] {$\displaystyle \Delta c$};
\draw (145.94,106.31) node [anchor=north west][inner sep=0.75pt]   [align=left] {$\displaystyle Q^{t+1}$};
\draw (368.41,203.05) node [anchor=north west][inner sep=0.75pt]   [align=left] {$\displaystyle Q^{t}$};

\end{tikzpicture}

\caption{The shadowed area is the feasible domain for values of $(Q^t,Q^{t+1})$.}
    \label{lemma1_fig}
\end{figure}
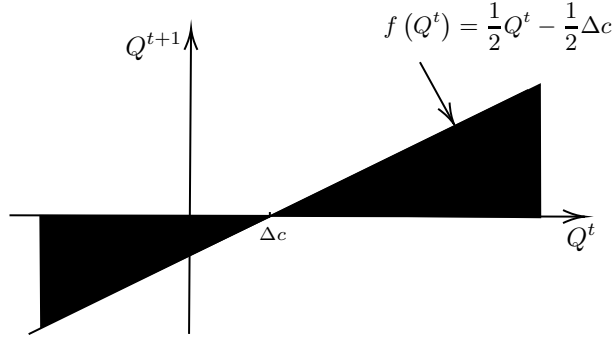

For the  case 
$  Q^t \leq \Delta c$,   it follows from Eq.~\eqref{Q-q} that $Q^{t+1}\leq 0  $, and then $Q^{t+k}\leq 0  $ for any integer $k\geq 1$ by induction. 
Otherwise,  that is, $  Q^t > \Delta c$, it follows from Eq.~\eqref{Q-q} that 
\begin{align*}
0\leq Q^{t+1}\leq  \frac{1 }{2} Q^t -\frac{ 1 }{2}\Delta c.
\end{align*}
We can always find a finite $t_0$ such that  $0\leq Q^{t+t_0}
\leq  \Delta c$. Then from the above case, we get
$ Q^{{t+t_0}+k}\leq 0   $  for any integer $k\geq 1$. 
To sum up,  there exists a positive integer $T$ such that $Q^t \leq  0$ for any $t>T$.

When $\Delta c=0$,  it follows from Eq.~\eqref{Q-q} that 
\begin{align*}
\min\left\{ 0, \frac{1 }{2} Q^t
 \right\}&\leq Q^{t+1}\leq  \max\left\{ 0,  \frac{1 }{2} Q^t \right\}
 \Rightarrow |Q^{t+1}| \le  \frac{1}{2}|Q^t|.
\end{align*}
It holds that $\lim_{t\rightarrow \infty}Q^t=0$.  
\end{proof}
Suppose the dynamic of best response does not converge to a Nash equilibrium, then the product quantity vector keeps changing. To study the changing pattern of product quantity, we start by focusing on the number of survival firms, which are firms that produce a positive output.
We use $m^t$ to represent the number of survival firms in round $t$.
We define
\begin{align*}
    \underline{n}=\lim_{t\rightarrow \infty}\inf\{m^k, k>t\}.
\end{align*}
Hence, there exists $T_1>T$ such that $\inf\{m^k,k>T_1\} = \underline{n}$.
Therefore, after round $T_1$, the first $\underline{n}$ firms would always survive by Lemma~\ref{lemma1}.

We focus on the dynamics after round $T_1$.
To make it simpler, we design a new game and a new instance of dynamics within it where we temporarily ignore the difference between $\underline{n}$ firms. In particular, we construct a new game where we use $\underline{n}$ identical firms to replace the first $\underline{n}$ possibly heterogeneous original firms. 
For $i\le \underline{n}$, firm $i$ has an identical cost 
$\underline{c}_i=\frac{1}{\underline{n}}\sum_{j=1}^{\underline{n}} c_j$.
For $i>\underline{n}$, firm $i$'s cost does not change, i.e., $\underline{c}_i=c_i$.
Then we design the following product quantity sequence $\underline{q}_i^t$ such that for any $t>0$,
$\underline{q}_i^t=\frac{1}{\underline{n}}\sum_{j=1}^{\underline{n}}q_j^{T_1+t}$ for $i\le \underline{n}$ and
$\underline{q}_i^t=q_i^{T_1+t}$ for $i>\underline{n}$. The following lemma shows that $\{\underline{\bm{q}}^t\}$ is still a valid product quantity sequence.

\begin{lemma}\label{lem:same_quantity}
Sequence $\{\underline{\bm{q}}^t\}$ is a product quantity sequence in best response dynamics in the new game. 
\end{lemma}
We need to show the product quantity $\{\underline{\bm{q}}^{t+1}\}$ is still a best response to $\{\underline{\bm{q}}^t\}$,
the idea is to divide the proof into two parts: one part addresses the case where $i>\underline{n}$, and the other part deals with the case where $i\leq \underline{n}$.

\begin{proof} 
We first show that for firm $i>\underline{n}$, it still makes the best response decision in the designed product quantity sequence. Since it chooses the best strategy in the sequence $\{\bm{q}^t\}$, we have
\begin{equation*}
    q_i^{t+T_1}=\frac12 \max \left\{ A-\sum_{j \neq i} q_j^{t+T_1}-c_i, 0\right\}.
\end{equation*}
According to the definition of $\underline{\bm{q}}^t$, we have
\begin{align*}
\underline{q}_i^t&=\frac{1}{2} \max \left\{ A-\sum_{j \leq \underline{n}} q_j^{t+T_1}-\sum_{j>\underline{n}, j \neq i} q_j^{t+T_1}-c_i, 0 \right\} \\
& =\frac{1}{2}  \max \left\{A-\sum_{j \leq \underline{n}} \underline{q}_j^t-\sum_{j > \underline{n}, j \neq i} \underline{q}_j^t-\underline{c}_i,0 \right\} \\
& =\frac{1}{2}  \max \left\{A-\sum_{j \neq i} \underline{q}_j^t-\underline{c}_i, 0\right\}.
\end{align*}
Hence, firm $i$ indeed uses the optimal strategy. 

Then, we consider the first $\underline{n}$ firms. For $i \leq \underline{n}$, since it uses the best response strategy in the sequence $\{\bm{q}^t\}$ and it always has a positive product quantity, we have 
\begin{equation*}
q_i^{t+T_1}=\frac{1}{2}\left(A-\sum_{j \neq i} q_j^{t+T_1}-c_i\right).
\end{equation*}
Sum up for all $i\leq \underline{n}$, we have
\begin{align*}
\sum_{i \leq \underline{n}} q_i^{t+T_1} & =\frac{1}{2} \left(\underline{n} A-\sum_{i \leq \underline{n}} c_i-\underline{n} \sum_{j > \underline{n}} q_j^{t+T_1}-(\underline{n}-1) \sum_{j \leq \underline{n}} q_j^{t+T_1}\right), \\
\underline{n}\underline{q}_i^t & =\frac{1}{2}\left(\underline{n}A-\underline{n} \underline{c}_i-\underline{n}\sum_{j>\underline{n}} \underline{q}_j^t -\underline{n} \sum_{j \neq i, j \leq \underline{n}} \underline{q}_j^t\right), \\
\underline{q}_i^t & =\frac{1}{2}\left(A-\sum_{j \neq i} \underline{q}_j^t -\underline{c}_i\right).
\end{align*}
Therefore, all firms are using the best response strategy in the new game.  
\end{proof}

When we focus on the number of survival firms after round $T_1$, the above lemma tells us that, w.l.o.g., we can assume the first $\underline{n}$ firms share the same cost and the same strategies.

Next, we show the firms' behavior will eventually converge to either a Nash equilibrium or an oscillation. 
We first consider the best response dynamics in the designed game $\{\underline{\bm{q}}^t\}$ and show there is a cyclical oscillation. The result is based on the cyclical variation in the production quantity of the first $\underline{n}$ firms. 
Then we extend the result to the best response dynamics in the original game.

\begin{lemma}\label{lemma3}
Both the original and new dynamics converge to either periodic oscillation or Nash equilibrium, and share the same length of oscillation period.
\end{lemma}
\begin{proof}
(1) We first consider the best response dynamics $\{\underline{\bm{q}}^t\}_{t \in \mathbb{N}}$ in the new game, by definition, there exist infinite $t$'s such that we have $m^t=\underline{n}$.
Among them, we can always choose $t_1<t_2$ such that $t_2-t_1$ is an even number. 
Otherwise, there must exist $t_1<t'<t_2$ such that $t_2-t'$ and $t'-t_1$ are both odd, we can choose $t_1$ and $t_2$ instead.

According to the new dynamic, the first \( \underline{n} \) entries of $\underline{\bm{q}}^t$ are identical, and the remaining entries are zero, so it must be that either \( \underline{\bm{q}}^{t_1} \leq \underline{\bm{q}}^{t_2} \), or \( \underline{\bm{q}}^{t_1} > \underline{\bm{q}}^{t_2} \) element-wisely. We will only consider the first case. The result for the case $\underline{\bm{q}}^{t_1}\ge \underline{\bm{q}}^{t_2}$ holds similarly. Next we show that if $\underline{\bm{q}}^{t_1}\le \underline{\bm{q}}^{t_2}$, that is, $\underline{q}^{t_1}_i\le \underline{q}^{t_2}_i$ for each $i\in[n]$, then for any even number $r$, we have $\underline{\bm{q}}^{t_1+r}\leq \underline{\bm{q}}^{t_2+r}$. By our definition, we have for each $i\in[n]$,
\begin{align*}
  \underline{q}_i^{t_1+1} = \max \left\{ \frac{1}{2}( A - \sum_{j\ne i}\underline{q}_j^{t_1} - c_i ), 0\right\}, 
\end{align*}
and
\begin{align*}
  \underline{q}_i^{t_2+1} = \max \left\{ \frac{1}{2}( A - \sum_{j\ne i}\underline{q}_j^{t_2} - c_i ), 0\right\}.
\end{align*}
Note that in this case, we have $\sum_{j\ne i}\underline{q}_j^{t_1}\le \sum_{j\ne i}\underline{q}_j^{t_2}$, so $\underline{q}_i^{t_1+1}\ge \underline{q}_i^{t_2+1}$ for any $i\in[n]$. By repeating the above process, we have that $\underline{q}_i^{t_1+2}\le \underline{q}_i^{t_2+2}$ for each $i\in[n]$.

Since $\Delta:=t_2-t_1$ is even, we have that the sequence
$\underline{\bm{q}}^{t_1}\le \underline{\bm{q}}^{t_1+\Delta}\le \underline{\bm{q}}^{t_1 + 2\Delta}\le \underline{\bm{q}}^{t_1 + 3\Delta}\le\cdots$.
Since $\underline{q}_i^t$ is bounded by $A$ for any $i$ and $t$, we must have 
the sequence $\{\underline{\bm{q}}^{ t_1+k \Delta} \}_{k\in \mathbb{N}}$ converges to a limit. 

Since the sequence $\{\underline{\bm{q}}^{t_1+k\Delta}\}_{k\in \mathbb{N}}$ converges, we have the sequence $\{\underline{\bm{q}}^{t_1+k\Delta+1}\}_{k\in \mathbb{N}}$ converges. It implies the sequence $\{\underline{\bm{q}}^{t_1+k\Delta+2}\}_{k\in \mathbb{N}}$ converges and so on. Therefore, for any $\ell<\Delta$, we have the sequence $\{\underline{\bm{q}}^{t_1+k\Delta+\ell}\}_{k\in \mathbb{N}}$ converges. We have now proven that, for the new dynamic $\{\underline{\bm{q}}^t\}_{t\in \mathbb{N}}$, if it does not converge to a Nash equilibrium\footnote{Nash equilibrium can also be regarded as an oscillation whose period equals to one.}, it will result in periodic oscillations.

(2) 
Since the new dynamics $\{\underline{\bm{q}}^t\}_{t\in \mathbb{N}}$ can converge to a periodic oscillation, it implies the total production of the first $\underline{n}$ firms and the production of the remaining firms can also converge to a periodic oscillation. The length of the oscillation period is the same as the one in the new dynamics, denoted by $\Delta$.
Hence, we have already shown that  $\sum_{i=1}^{\underline{n}}q_{i}^t=\sum_{i=1}^{\underline{n}}q_{i}^{t+\Delta}$ and $q_{j}^t=q_{j}^{t+\Delta}$, for any index $j\in [\underline{n}+1,n]$. It is left to prove that the production of each of the first $\underline{n}$ firms converges to a periodic oscillation with length $\Delta$.

We consider the first firm. For the production of the first firm at time $t>T_1$ and $t+\Delta$, if $q_{1}^t$ and $q_{1}^{t+\Delta}$ are not equal, there are two cases to consider. We only provide the proof for the case $q_{1}^t<q_{1}^{t+\Delta}$. The case where $q_{1}^t>q_{1}^{t+\Delta}$ follows similarly. According to the the best response dynamics, we have:
\begin{align*}
  q_{1}^{t+1} = \frac{1}{2}( A - \sum_{i=1}^{\underline{n}}q_{i}^t- \sum_{j=\underline{n}+1}^{n}q_{j}^t - c_1 + q_{1}^t ), 
\end{align*}
and
\begin{align*}
  q_{1}^{t+\Delta+1} = \frac{1}{2}( A - \sum_{i=1}^{\underline{n}}q_{i}^{t+\Delta}- \sum_{j=\underline{n}+1}^{n}q_{j}^{t+\Delta} - c_1 + q_{1}^{t+\Delta} ).
\end{align*}
Subtract these two equations, we get
\begin{align*}
   q_{1}^{t+1}- q_{1}^{t+\Delta+1} = \frac12(q_{1}^{t}- q_{1}^{t+\Delta}).
\end{align*}

Thus we have $q_{1}^{t+1}<q_{1}^{t+\Delta+1}$. For any number $r$, it follows that $q_{1}^{t+r}<q_{1}^{t+\Delta+r}$. Furthermore, we obtain the inequality chain
$q_{1}^t< q_{1}^{t+\Delta}< q_{1}^{t+2\Delta}< q_{1}^{t+3\Delta}<\cdots$. As production can not be arbitrarily large, the sequence
$\{q_{1}^{t+k\Delta}\}_{k\in \mathbb{N}}$ converges to a limit.  

Since the sequence $\{q_1^{t+k\Delta}\}_{k\in \mathbb{N}}$ converges, the sequence $\{q_1^{t+k\Delta +1}\}_{k\in \mathbb{N}}$ must also converge. This implies that the sequence $\{q_1^{t+k\Delta +2}\}_{k\in \mathbb{N}}$ converges, and so on. Thus, for any $l<\Delta$, the sequence $\{q_1^{t+k\Delta +l}\}_{k\in \mathbb{N}}$ converges as well. Then the sequence $\{q_1^{t}\}_{t\in \mathbb{N}}$ converges to an oscillation and the length of the oscillation is a factor of $\Delta$.
By the same reasoning, the production levels of the other $\underline{n}-1$ firms will also converge. The length of each cyclical oscillation is always a factor of $\Delta$.
Combining the remaining firms except the first $\underline{n}$ firms, we can conclude that the original dynamics $\{\bm{q}^{t} \}_{t\in \mathbb{N}}$ converge to a periodic oscillation with length $\Delta$ or a Nash equilibrium.  
\end{proof}

Oscillations of period two have been observed in the literature. However, no oscillation of longer periods is observed or studied. Our main theorem shows that there is no oscillation of period greater than two.

\begin{theorem}\label{theorem2}
The period of oscillation cannot be greater than two.
\end{theorem}
According to Lemma~\ref{lemma3}, the original and new dynamics share the same period. Therefore, we will apply  techniques similar to those in Lemma~\ref{lem:same_quantity} to construct the new dynamic, with particular attention to analyzing this periodicity.
For a better illustration of the dynamics, we introduce a matrix to represent the quantity of each firm in each round.  
Given $n,t>0$, we say a matrix $M\in\mathbb{R}^{t\times n}_+$ is a {\em $t$-period quantity matrix} for $n$ firms, if $M_{i^{+1},j}$ is the updated quantity of $(M_{i,1},\ldots,M_{i,n})$ for firm $j$, where $i^{+1}:=(i+1) \pmod t + 1$. 
That is, all firms update their production quantities by using the previous row of the matrix. Here is an example of the $2$-period quantity matrix. There are four firms with $A=20$ and $c_1=c_2=c_3=c_4=0$. There is a $2$-period oscillation in the dynamics described by the following $2$-period quantity matrix. 
\begin{align*}
\begin{bmatrix}
    10& 10 & 10 &10\\
    0&  0 & 0 & 0
\end{bmatrix}.
\end{align*}
\begin{proof}

We will conduct a detailed analysis on a case-by-case basis, focusing on two cases based on the status of the surviving firms.

In the first case, we examine the situation where the number of surviving firms fluctuates. If a group of firms simultaneously supply positive quantities, they can be substituted with an equivalent number of identical firms, each with their average cost. 
This substitution can be justified using similar arguments to those presented in Lemma ~\ref{lem:same_quantity}. 
To streamline our discussion, we use abbreviations for the quantity matrices, retaining only one element to represent all identical firms (those with the same quantity and cost).

In the second case, we consider the case where the number of surviving firms remains constant. Similar to the first case, we use abbreviations for the quantity matrices to simplify our discussion.


\paragraph{\bf Case 1} Based on the assumption in this case, the matrix must contain at least one row of the form $[x\ 0\ 0\ 0\ \dots\ 0]$. Given that this matrix forms a cycle, we can rearrange it so that a row of the form $[x\ 0\ 0\ 0\ \dots\ 0]$ becomes the first row. If there are multiple rows of this form, we choose the one with the smallest value of $x$ as the first row. We denote the set of indices for these rows as $\mathcal{F}$. After this adjustment, the matrix can be represented as:
\begin{align*}
\begin{bmatrix}
    a_{1,1}&0&\cdots&0\\
    a_{2,1}&a_{2,2}&\cdots&a_{2,n}\\
    a_{3,1}&a_{3,2}&\cdots&a_{3,n}\\
    \vdots & \vdots  & & \vdots   \\
    a_{t,1}&a_{t,2}&\cdots&a_{t,n}\\
\end{bmatrix}.
\end{align*}
Here, the first row is $[x\ 0\ 0\ 0\ \dots\ 0]$, and the remaining rows contain non-negative elements. We will now proceed with a case-by-case analysis. Our objective is to demonstrate that after this rearrangement, the first row maintains a partial order relation with all other rows, implying it is less than or equal to each of the other rows:
\begin{align*}
[a_{1,1}\ 0\ 0\ \dots\ 0]\leq[a_{i,1}\ a_{i,2}\ a_{i,3}\ \dots\ a_{i,n}].
\end{align*}


When $a_{1,1} < 0$, it will not occur because the production quantity cannot be negative. When $a_{1,1} = 0$, it holds obviously.
While for the case that $a_{1,1} > 0$: 
We express the first two columns of the matrix using the update formula Eq.~\eqref{q_i^{t+1}}. For example, we can list the following relation for the first row, where $k_i$  represents the number of firms in the $k$-th category:

\begin{align}
a_{1,1}&=\frac{A-(k_1-1)a_{t,1}-k_2a_{t,2}-\dots -k_na_{t,n}-c_1}{2}\label{eq:1},\\
a_{1,2}&=0\ge\frac{A-k_1a_{t,1}-(k_2-1)a_{t,2}-\dots -k_na_{t,n}-c_2}{2}\label{eq:2}.
\end{align}

Subtracting adjacent equations, such as subtracting Eq.~\eqref{eq:2} from Eq.~\eqref{eq:1}, yields the following inequalities:

\begin{align}
a_{1,1}-a_{1,2}&\leq\frac{a_{t,1}-a_{t,2}+c_2-c_1}{2}\nonumber.
\end{align}

Similarly, we can write the expressions for the first two columns of the remaining rows after taking the differences as follows:
\begin{align}
a_{2,1}-a_{2,2}&\leq\frac{a_{1,1}-a_{1,2}+c_2-c_1}{2}\nonumber,\\
&\cdots\nonumber \\
a_{t,1}-a_{t,2}&\leq\frac{a_{t-1,1}-a_{t-1,2}+c_2-c_1}{2}\nonumber.
\end{align}

Summing these inequalities results in:
\begin{align}
(a_{1,1}-a_{1,2})+(a_{2,1}-a_{2,2})+\cdots+(a_{t,1}-a_{t,2})\leq t(c_2-c_1)\label{eq:9}.
\end{align}


Next, we examine the relationship between each $(a_{i,1} - a_{i,2})$ and $(c_2 - c_1)$.
In this context, the number of surviving firms changes in each turns. There exists an index \(i \in \mathcal{F}\) such that \((i+1) \not\in \mathcal{F}\) and \((i-1)^* \in [t]\). Consequently, the matrix contains the following three adjacent rows:
\begin{align*}
\begin{bmatrix}
    a_{(i-1)^*,1}&a_{(i-1)^*,2}&\cdots&a_{(i-1)^*,n}\\
    a_{i,1}&0&\cdots&0\\
    a_{i+1,1}&a_{i+1,2}&\cdots&a_{i+1,n}
\end{bmatrix},
\end{align*}
where $(i-1)^*$ represents the case where if $i=1$, then $(i-1)^*=t$; and $(i-1)^*$ remains unchanged as $i-1$ otherwise.

For the last two rows of expressions, by taking the difference of the first two columns, we derive the following expressions:
\begin{align}
(a_{i,1}-0)-(c_2-c_1)&\leq\frac{(a_{(i-1)^*,1}-a_{(i-1)^*,2})-(c_2-c_1)}{2}\label{eq:10},\\
(a_{i+1,1}-a_{i+1,2})-(c_2-c_1)&=\frac{(a_{i,1}-0)-(c_2-c_1)}{2}.\label{eq:11}
\end{align}

It can be observed that $(a_{i,1}-0)$ and $(a_{i+1,1}-a_{i+1,2})$ are either both less than $(c_2-c_1)$ or both greater than $(c_2-c_1)$. 
More generally, if $i \in \mathcal{F}$, $i < j$, and $i+1, \dots, j \not\in \mathcal{F}$, then $(a_{i,1} - 0)$ and $(a_{j,1} - a_{j,2})$ are either both less than $(c_2 - c_1)$ or both greater than $(c_2 - c_1)$.
Therefore, the problem reduces to analyzing the relationship between $(a_{i,1}-0)$ and $(c_2-c_1)$, where $i\in \mathcal{F}$. 
Additionally, if $(a_{i,1}-0)<(c_2-c_1)$, then 
\begin{align}
(a_{i,1}-0)<(a_{i+1,1}-a_{i+1,2})<\cdots<(a_{j,1}-a_{j,2})<(c_2-c_1)\label{eq:12}.
\end{align}
If $(a_{i,1}-0)>(c_2-c_1)$, then 
\begin{align}
(a_{i,1}-0)>(a_{i+1,1}-a_{i+1,2})>\cdots>(a_{j,1}-a_{j,2})>(c_2-c_1)\label{eq:133}.
\end{align}
And if $(a_{i,1}-0)=(c_2-c_1)$, then 
\begin{align}
(a_{i,1}-0)=(a_{i+1,1}-a_{i+1,2})=\cdots=(a_{j,1}-a_{j,2})=(c_2-c_1)\label{eq:13}. 
\end{align}

Based on the placement of the smallest value of $x$ in the first row, we have $a_{1,1}\leq a_{i,1}$ for $i \in \mathcal{F}$. From Eq.~\eqref{eq:9}, it follows that $(a_{1,1}-0)\leq (c_2-c_1)$. 
Additionally, utilizing Eq.~\eqref{eq:12}, Eq.~\eqref{eq:133} and Eq.~\eqref{eq:13}, we deduce that $a_{1,1}\leq a_{j,1}$ where $j\in [t]$.

This allows us to establish a partial order relationship between the first row and any other row:
\begin{align}
[a_{1,1}\ 0\ 0\ \dots\ 0]\leq[a_{j,1}\ a_{j,2}\ a_{j,3}\ \dots\ a_{j,n}]\label{eq:14}.
\end{align}

Referring to Eq.~\eqref{eq:14} and Eq.~\eqref{q_i^{t+1}}, we further establish that the second row holds a partial order relationship with the other rows:
\begin{align}
[a_{2,1}\ a_{2,2}\ a_{2,3}\ \dots\ a_{2,n}]\ge[a_{j,1}\ a_{j,2}\ a_{j,3}\ \dots\ a_{j,n}]\label{eq:15}.
\end{align}

Similarly, based on Eq.~\eqref{eq:15} and Eq.~\eqref{q_i^{t+1}}, we find that the third row has a partial order relationship with the other rows:
\begin{align}
[a_{3,1}\ a_{3,2}\ a_{3,3}\ \dots\ a_{3,n}]\leq[a_{j,1}\ a_{j,2}\ a_{j,3}\ \dots\ a_{j,n}]\label{eq:16}.
\end{align}

Combining to Eq.~\eqref{eq:14} and Eq.~\eqref{eq:16}, we observe that the first row is equivalent to the third row:
\begin{align}
[a_{1,1}\ 0\ 0\ \dots\ 0]=[a_{3,1}\ a_{3,2}\ a_{3,3}\ \dots\ a_{3,n}]\label{eq:17}.
\end{align}


Hence, under these conditions, when the number of surviving firms fluctuates, only a two-cycle oscillation can occur. Since the original and new dynamics share the same period, fluctuations in the number of surviving firms will cause the original dynamic to exhibit only a two-cycle oscillation.


\paragraph{\bf Case 2}
When the number of surviving firms remains constant, we cannot estimate the firms' output and cost using the aforementioned averaging method. Instead, we represent each firm individually with a column in the following matrix:
\begin{align*}
\begin{bmatrix}
    a_{1,1}&a_{1,2}&\cdots&a_{1,n}\\
    a_{2,1}&a_{2,2}&\cdots&a_{2,n}\\
    a_{3,1}&a_{3,2}&\cdots&a_{3,n}\\
    \vdots & \vdots  & & \vdots   \\
    a_{t,1}&a_{t,2}&\cdots&a_{t,n}\\
\end{bmatrix},    
\end{align*}
where $a_{i,j}>0$, $k_j=1$, $\forall i \in [n]$ and $\forall i \in [t]$

We continue by applying the updating formula to the first two columns and then subtracting the corresponding elements in each row, yielding the following results:
\begin{align}
a_{1,1}-a_{1,2}&=\frac{a_{t,1}-a_{t,2}+c_2-c_1}{2}\nonumber ,\\
a_{2,1}-a_{2,2}&=\frac{a_{1,1}-a_{1,2}+c_2-c_1}{2}\nonumber ,\\
&\cdots\nonumber \\
a_{t,1}-a_{t,2}&=\frac{a_{t-1,1}-a_{t-1,2}+c_2-c_1}{2}\nonumber .
\end{align}
Simplifying, we get:
\begin{align}
a_{1,1}-a_{1,2}=a_{2,1}-a_{2,2}=\cdots=a_{t,1}-a_{t,2}=c_2-c_1\label{eq:18}.
\end{align}

For the first column of the matrix, without loss of generality, let $a_{1,1}\leq a_{i,1}$,  where $i\in [t]$.

According to Eq.~\eqref{eq:18}, it follows that $a_{1,j}\leq a_{i,j}$, where $i\in [t]$ and $j\in [n]$.

Applying a similar approach as before, we can establish a partial order relationship between the first row and the remaining rows:
\begin{align}
[a_{1,1}\ a_{1,2}\ a_{1,3}\ \dots\ a_{1,n}]\leq[a_{i,1}\ a_{i,2}\ a_{i,3}\ \dots\ a_{i,n}]\label{eq:19}.
\end{align}

From Eq.~\eqref{eq:19} and Eq.~\eqref{q_i^{t+1}}, we can deduce that the second row has a partial order relationship with the other rows:
\begin{align}
[a_{2,1}\ a_{2,2}\ a_{2,3}\ \dots\ a_{2,n}]\ge[a_{i,1}\ a_{i,2}\ a_{i,3}\ \dots\ a_{i,n}]\label{eq:20}.
\end{align}

Similarly, using to Eq.~\eqref{eq:20} and Eq.~\eqref{q_i^{t+1}}, we can establish that the third row has a partial order relationship with the other rows: 
\begin{align}
[a_{3,1}\ a_{3,2}\ a_{3,3}\ \dots\ a_{3,n}]\leq[a_{i,1}\ a_{i,2}\ a_{i,3}\ \dots\ a_{i,n}]\label{eq:21}.
\end{align}

Combining Eq.~\eqref{eq:19} and Eq.~\eqref{eq:21}, we find that the first row is equal to the third row:
\begin{align}
[a_{1,1}\ 0\ 0\ \dots\ 0]=[a_{3,1}\ a_{3,2}\ a_{3,3}\ \dots\ a_{3,n}]\label{eq:22}.
\end{align}

Therefore, in this case, when the number of surviving firms remains unchanged, only a two-cycle oscillation exists.  
\end{proof}

Next, we study the forms for the two-period oscillation. We first classify by whether the number of companies producing products in both periods, $k_1$, is zero or not. For \( k_1 > 0 \), we further classify the forms based on the number of companies producing products only in one of the two periods. We investigate the time complexity of verifying whether a given form of oscillation exists and the time complexity of finding an oscillation in that form if it exists. 

\begin{theorem}
\label{thm:format}
Let  $k_1$ denote the number of firms with consistently positive production during a two-period stable oscillation. Depending on the value of $k_1$, the system exhibits three distinct equilibrium types:
\begin{enumerate}
    \item \textbf{Case  $k_1 = 0$:} In one of the two periods, all firms have zero production. 
    \begin{itemize}
        \item Equilibrium existence can be verified in  $O(n)$ time.
        \item The unique equilibrium can be computed in $O(n)$ time.
    \end{itemize}
    
    \item \textbf{Case $0 < k_1 < 3$:} Exactly $k_1$ firms maintain positive production in both periods, while all others have zero production in one period.
    \begin{itemize}
        \item Oscillation existence can be verified in $O(n)$ time.
        \item Infinitely many oscillations may exist.
        \item A solution can be find in $O(n)$ time.
    \end{itemize}
    
    \item \textbf{Case  $k_1 \geq 3$ :} Stable oscillations emerge among three or more firms.
    \begin{itemize}
        \item Oscillation existence can be verified in $O(1)$ time.
        \item Infinitely many oscillations may exist.
        \item A solution can be find in $O(1)$ time.
    \end{itemize}
\end{enumerate}
\end{theorem}

\begin{proof}
We assume that there are $k$ firms whose production could be positive when oscillations occur, and the production costs of the firms are  $c_1, \ldots, c_n$. We denote by $A$ the market capacity. 
Next, we will separately discuss the constraints on $k$ for these cases, as well as the time complexity for checking the existence of oscillations and finding the oscillation solutions.

\paragraph{\bf When $k_1=0$} This corresponds to Case 1 where no firm survives in the first period, and there are $k$ survival firms in the second period. It can be intuitively represented using the quantity matrix as follows: 
\begin{align*}
\begin{bmatrix}
0&0&\cdots&0\\
a_{2,1}&a_{2,2}&\cdots&a_{2,k}\\
\end{bmatrix}.
\end{align*}
If the equilibrium exists, the value of $k$ can be determined in $O(n)$ time. Since $a_{2,k}$ is the last positive number in the second period, the following inequalities can be established:
\begin{align}
a_{2,k}=\frac{A-c_k}{2}&>0,\label{Theorem2-case1--1}\\
\frac{A-c_{k+1}}{2}&\leq0.\label{Theorem2-case1-0}
\end{align}
By the update formula, the oscillation solution can be calculated in $O(n)$ time:
\begin{align}
a_{2,i}=\frac{A-c_i}{2}.\label{Theorem2-case1-1}
\end{align}
Since no firm survives in the first stage (i.e., \( k_1 = 0 \)), according to the update formula, we can list the following relationship for the first firm in the first period:
\begin{align}
a_{1,1}=\frac{A-\sum_{i=1}^{k} a_{2,i}-c_i+a_{2,1}}{2}\leq 0.\label{Theorem2-case1-2}
\end{align}
By substitute Eq.~\eqref{Theorem2-case1-1} into Eq.~\eqref{Theorem2-case1-2}, we have
\begin{align}
(k-3)A+3c_1-\sum_{i=1}^{k}c_i\ge 0.\label{Theorem2-case1-3}
\end{align}
Therefore, the solution that satisfies Case 1 can be verified in $O(n)$ time complexity according to Eq.~\eqref{Theorem2-case1--1}, Eq.~\eqref{Theorem2-case1-0} and Eq.~\eqref{Theorem2-case1-3}.
    

\paragraph{\bf When $k_1>0$ and $k_2>0$} W.l.o.g., we assume there are $k_1>0$ survival firms in the first period, and $k_1 + k_2=k$ survival firms in the second period, where $k_2>0$. It can be intuitively represented using the quantity matrix as follows: 

\begin{align*}
    \begin{bmatrix}
    a_{1,1}&a_{1,2}&\cdots&a_{1,k_1}&0&0&\cdots&0\\
    a_{2,1}&a_{2,2}&\cdots&a_{2,k_1}&a_{2,k_1+1}&a_{2,k_1+2}&\cdots&a_{2,k_1+k_2}\\
    \end{bmatrix}.
\end{align*}
Similarly to the arguments in Lemma~\ref{lem:same_quantity}, we can replace the firms simultaneously supplying positive quantities with a set of equivalent firms. In particular, we assume that in the first period, $k_1$ firms have the average production quantity of $a$ and $k_2$ firms have a production quantity of 0. In the second period, $k_1$ firms have the average production quantity of $b$, and $k_2$ firms have a average production quantity of $c$. The average production cost per product for the $k_1$ firms are $x$, for the $k_2$ firms' average production cost are $y$, and it is easy to see that $x<y$. The matrix representation can be simplified to:

\begin{align*}
    \begin{bmatrix}
    a&0\\
    b&c\\
    \end{bmatrix},
\end{align*}
where $a>0, b>0$ and $c>0$. Since the number of firms that consistently survive in the original dynamic is the same as that in the new dynamic, we will now calculate the range of firms that consistently survive in the new dynamic.

According to the update formula, we have
\begin{align}
a=&\frac{A-\left(k_1-1\right) b-k_2 c-x}{2}, \label{theorem3_eq1}\\
b=&\frac{A-\left(k_1-1\right) a-x}{2}, \label{theorem3_eq2}\\
c=&\frac{A-k_1 a-y}{2}, \label{theorem3_eq3}\\
0 \geqslant& \frac{A-k_1 b-\left(k_2-1\right) c-y}{2}.\label{theorem3_eq4}
\end{align}

If we update certain parameters to the following form:
$A^{\prime}=A-x>0$, $y^{\prime}=y-x>0$ and $x^{\prime}=0$, we have Eqs.~\eqref{theorem3_eq1} to \eqref{theorem3_eq4} still hold. Therefore, we assume $x=0$ without loss of generality.
Then, by substituting Eqs.~\eqref{theorem3_eq2} and \eqref{theorem3_eq3} into Eq.~\eqref{theorem3_eq1}, we obtain
\begin{align}
&\left(k_1+k_2-3\right) A =\left(k_1^2-2 k_1-3+k_1 k_2\right) a+k_2 y.\label{theorem3_eq5}
\end{align}

Eq. \eqref{theorem3_eq3}  can be deformed to obtain 
\begin{align}
    A>k_1 a+y. \label{theorem3_eq6}
\end{align}

Pluging Eq. \eqref{theorem3_eq2} and Eq. \eqref{theorem3_eq3} into Eq. \eqref{theorem3_eq4} yields
\begin{align}
&\left(k_2-3\right) y+k_1\left(k_1+k_2-2\right) a \leq\left(k_1+k_2-3\right) A. \label{theorem3_eq7}
\end{align}

Combining Eq.~\eqref{theorem3_eq5} with Eq.~\eqref{theorem3_eq7}, we obtain:
\begin{align*}
     a \leq y.
\end{align*}

Combining Eq.~\eqref{theorem3_eq5} with Eq.~\eqref{theorem3_eq6}, we obtain:
\begin{align*}
     (k_1-3)a > (k_1-3)y.
\end{align*}

Therefore, in this case, the number of firms that consistently survive in the new dynamic is less than 3. Since the number of firms that consistently survive in the original dynamic is the same as in the new dynamic, \( k_1 \) in the original dynamic must also be less than 3.

\paragraph{\bf When $0<k_1<3$ and $k_2>0$} It corresponds to Case 2. 
The number of firms with consistently positive production, $k_1$, is 1 or 2. It can be intuitively represented using matrices as follows: 
\begin{align*}
\begin{bmatrix}
a_{1,1}&0&0&\cdots&0\\
a_{2,1}&a_{2,2}&a_{2,3}&\cdots&a_{2,k}\\
\end{bmatrix} \text{~~and~~}
\begin{bmatrix}
a_{1,1}&a_{1,2}&0&0&\cdots&0\\
a_{2,1}&a_{2,2}&a_{2,3}&a_{2,4}&\cdots&a_{2,k}
\end{bmatrix}.
\end{align*}
Next, we calculate the specific form of the final solution if an equilibrium represented by the case exists. It can be expressed in a general matrix form as follows:
\begin{align*}
\begin{bmatrix}
a_{1,1}&a_{1,2}&\cdots&a_{1,k_1}&0&0&\cdots&0\\
a_{2,1}&a_{2,2}&\cdots&a_{2,k_1}&a_{2,k_1+1}&a_{2,k_1+2}&\cdots&a_{2,k_1+k_2}\\
\end{bmatrix}.
\end{align*}

According to the update formula, the expression for the second stage can be written as:
\begin{align}
a_{2,l}&=\frac{1}{2}\left(A-\sum_{i=1}^{k_1} a_{1,i}-c_l+a_{1,l}\right), \quad \forall 1\leq l\leq k_1, \label{Theorem3-case2-32} \\
a_{2,k_{1}+r}&=\frac{1}{2}\left(A-\sum_{i=1}^{k_1} a_{1,i}-c_{k_1+r}\right), \quad \forall 1\leq r\leq k_2. \label{Theorem3-case2-33}
\end{align}

Summing both sides of the above equation, we obtain:
\begin{align}
\sum_{i=1}^{k_1} a_{2,i}&=\frac{1}{2}\left(k_1 A-k_1 \sum_{i=1}^{k_1} a_{1,i}-\sum_{i=1}^{k_1} c_i+\sum_{i=1}^{k_1} a_{1,i}\right),\label{Theorem3-case2-34}\\
\sum_{j=k_1+1}^{k_1+k_2} a_{2,j}&=\frac{1}{2}\left(k_2 A-k_2 \sum_{i=1}^{k_1} a_{1,i}-\sum_{j=k_1}^{k_1+k_2} c_j\right).\label{Theorem3-case2-35}
\end{align}

Since the system satisfies a two-period oscillation, the quantity in the first stage can be calculated using the quantity in the second stage. This can be expressed using the update formula as follows:
\begin{align}
a_{1,l}=\frac{1}{2}\left(A-\sum_{i=1}^{k_1} a_{2,i}-\sum_{j=k_1+1}^{k_1+k_2} a_{2,j}-c_l+a_{2,l}\right), \quad \forall 1\leq l\leq k_1. \label{Theorem3-case2-36} 
\end{align}

We substitute Eq.~\eqref{Theorem3-case2-32} to Eq.~\eqref{Theorem3-case2-35}  into Eq.~\eqref{Theorem3-case2-36}  and after simplification we have:
\begin{align}
3 a_{1,l}=&\left(3-k_1-k_2\right) A+\left(k_1+k_2-2\right) \sum_{i=1}^{k_1} a_{1,i}\notag\\ &+\left(\sum_{i=1}^{k_1} c_i+\sum_{j=k_1+1}^{k_1+k_2} c_j\right)-3 c_l, \quad \forall 1\leq l\leq k_1. \label{Theorem3-case2-37}
\end{align}

By summing up Eq.~\eqref{Theorem3-case2-37}, we have
\begin{align}
3 \sum_{i=1}^{k_1} a_{1,i}=&\left(3-k_1-k_2\right) k_1 A+ ( k_1+k_2-2) k_1 \sum_{i=1}^{k_1} a_{1,i}\notag\\
&+k_1\left(\sum_{i=1}^{k_1} c_i+\sum_{j=k_1+1}^{k_1+k_2} c_j\right)-3 \sum_{i=1}^{k_1} c_i.\label{Theorem2-case2-101}
\end{align}

If there exists an oscillation with  $k_1 = 1$, by iterating over the values of $k_2$ and using Eq.~\eqref{Theorem2-case2-101}, we can calculate the following equilibrium solution:

If $k_2 \neq 4$ we have:
\begin{align*}
a_{1,1}&=\frac{(2-k_2)A+\sum_{j=2}^{1+k_2} c_j-2c_1}{4-k_2},\\
a_{2,1}&=\frac{A-c_1}{2},\\
a_{1,j}&=0, \quad  \forall 2\leq j\leq 1+k_2,\\
a_{2,j}&=\frac{A-a_{1,1}-c_j}{2}, \quad  \forall 2\leq j\leq 1+k_2.
\end{align*}

If $k_2 = 4$, since the coefficient of $a_{1,1}$ is zero, an infinite number of oscillation exist. Based on the oscillation form, the range can be calculated and represented as follows:
\begin{align*}
A-c_6\leq &a_{1,1}<A-c_5,\\
&a_{2,1}=\frac{A-c_1}{2},\\
&a_{1,j}=0,\quad  \forall 2\leq j\leq 1+k_2,\\
&a_{2,j}=\frac{A-a_{1,1}-c_j}{2}, \quad  \forall 2\leq j\leq 1+k_2.
\end{align*}

If there exists an oscillation with  $k_1 = 2$, by iterating over the values of $k_2$ and using equation Eq.~\eqref{Theorem2-case2-101}, we can calculate the following equilibrium solution:
\begin{align*}
a_{1,1}&=\frac{(1-k_2)A+k_2\frac{2(1-k_2)A+2\sum_{j=3}^{2+k_2}c_j-(c_1+c_2)}{3-2k_2}+\sum_{j=3}^{2+k_2}c_j+c_2-2c_1}{3},\\
a_{1,2}&=\frac{(1-k_2)A+k_2\frac{2(1-k_2)A+2\sum_{j=3}^{2+k_2}c_j-(c_1+c_2)}{3-2k_2}+\sum_{j=3}^{2+k_2}c_j+c_1-2c_2}{3},\\
a_{2,1}&=\frac{A-a_{1,2}-c_1}{2},\\
a_{2,2}&=\frac{A-a_{1,1}-c_2}{2},\\
a_{1,j}&=0,\quad  \forall 3\leq j\leq 2+k_2,\\
a_{2,j}&=\frac{A-a_{1,1}-a_{1,2}-c_j}{2},\quad  \forall 3\leq j\leq 2+k_2.
\end{align*}
Therefore, the oscillation solution can be calculated in $O(n)$ time. According to the update formula, we can verify whether this equilibrium holds in $O(n)$ time.

\paragraph{\bf{When $k_1> 0$ and $k_2=0$}} It corresponds to Case 3 where the number of survival firms unchanged. It can be intuitively represented as follows: 
    \begin{align*}
    \begin{bmatrix}
    a_{1,1}&a_{1,2}&\cdots&a_{1,k}\\
    a_{2,1}&a_{2,2}&\cdots&a_{2,k}\\
    \end{bmatrix}.
    \end{align*}

We derive the following results:
\begin{align}
& a_{1,i}=\frac{A-\sum_{j \neq i} a_{2,j}-c_i}{2}, \label{theorem3_eq8}\\
& a_{2,i}=\frac{A-\sum_{j \neq i} a_{1,j}-c_i}{2}, \label{theorem3_eq9}
\end{align}

Subtracting Eq.~\eqref{theorem3_eq9} from Eq.~\eqref{theorem3_eq8} and simplifying, we obtain:
\begin{align}
& 3\left(a_{1,i}-a_{2,i}\right)=\sum_{j=1}^k a_{1,j}-\sum_{j=1}^k a_{2,j}, \quad \forall i. \label{theorem3_eq10}
\end{align}

Hence we have that $a_{1,i}-a_{2,i}$ is a constant and $k=3$. At this point, the quantity matrix can be simplified as follows:
\begin{align*}
\begin{bmatrix}
a_{1,1}&a_{1,2}&a_{1,3}\\
a_{2,1}&a_{2,2}&a_{2,3}
\end{bmatrix}.
\end{align*}

We denote $\Delta a=a_{2,i}-a_{1, i}$ for any $i\in[n]$ and WLOG assume that $\Delta a>0$.
By the update formula and simplifications, we get
\begin{align}
a_{1,1}&=\frac{A-a_{1,2}-a_{1,3}-2\Delta a -c_1}{2},\notag\\
a_{1,1}+c_1&=A-(a_{1,1}+a_{1,2}+a_{1,3})-2\Delta a.\label{Theorem2-case3-1}
\end{align}
Similarly, we can derive the expressions for $a_{1,2}$ and $a_{1,3}$ and observe the following pattern:
\begin{align}
a_{1,1}+c_1=a_{1,2}+c_2=a_{1,3}+c_3.\label{Theorem2-case3-2}
\end{align}
According to Eq.~\eqref{Theorem2-case3-2}, we rewrite and simplify Eq.~\eqref{Theorem2-case3-1} as follows:

\begin{align*}
a_{1,1}+c_1&=A-\left(a_{1,1}+a_{1,1}+c_1-c_2+a_{1,1}+c_1-c_3\right)-2\Delta a,\\
a_{1,1}&=\frac{1}{4}\left(A+c_2+c_3-3 c_1\right)-\frac{1}{2} \Delta a.
\end{align*}
Similarly, we can obtain the followings:
\begin{align*}
& a_{1,1}=\frac{1}{4}\left(A+c_2+c_3-3 c_1\right)-\frac{1}{2} \Delta a, \\
& a_{1,2}=\frac{1}{4}\left(A+c_1+c_3-3 c_2\right)-\frac{1}{2} \Delta a, \\
& a_{1,3}=\frac{1}{4}\left(A+c_1+c_2-3 c_3\right)-\frac{1}{2} \Delta a, \\
& a_{2,1}=\frac{1}{4}\left(A+c_2+c_3-3 c_1\right)+\frac{1}{2} \Delta a,  \\
& a_{2,2}=\frac{1}{4}\left(A+c_1+c_3-3 c_2\right)+\frac{1}{2} \Delta a,\\
& a_{2,3}=\frac{1}{4}\left(A+c_1+c_2-3 c_3\right)+\frac{1}{2} \Delta a.
\end{align*}
Since $a_{1,3}$ and $a_{2,3}$ must be non-negative and the fourth firm does not produce any product, we have
\begin{align*}
a_{1,3}=\frac{1}{4}\left(A+c_1+c_2-3 c_3\right)-\frac{1}{2} \Delta a&>0,\\
\frac{A-a_{1,1}-a_{1,2}-a_{1,3}-c_4}{2}&\leq0.
\end{align*}
We simplify the above and get
\begin{align*}
& \left\{\begin{array}{l}
\Delta a<\frac{A+c_1+c_2-3 c_3}{2},\\
\Delta a\leq \frac{2}{3}c_4-\frac{A+c_1+c_2+c_3}{6}.
\end{array}\right.
\end{align*}
Thus, we have found the range of the production difference for the same company over two periods.
\end{proof}

\section{Conclusion}
In this paper, we characterized specific patterns in the oscillations of production quantities, where firms' production decisions alternate between two distinct periods. This behavior arises due to the complexity of firms' best response strategies, especially when ensuring they produce non-negative quantities. 
We consider all three specific types of two-period oscillations. Each type can be found efficiently, taking linear time in the number of firms involved. This classification enhances our understanding of strategic interactions among firms producing homogeneous products, providing insights into the intricate dynamics of markets with multiple competitors. Furthermore, our techniques are robust and can be extended to other settings. In particular, Theorem~\ref{theorem2} still holds when firms use the weighted average between the best response and the strategy in the previous round. 
 
For the future work, it would be interesting to examine the convergence rate to reach an oscillation and the stability of different forms of oscillations. Besides, how to extend our results to Bertrand model~\cite{bertrand1883review} could be a challenging problem. In Bertrand model, each firm competes by setting prices instead of quantities.

\section*{Acknowledgments}
This work is supported by the National Natural Science Foundation of China (Nos. 62472029 and 62172422) and the Key Laboratory of Interdisciplinary Research of Computation and Economics (Shanghai University of Finance and Economics), Ministry of Education. 

%
%
%
\newpage

\bibliographystyle{plain}
\bibliography{ref}

\end{document}